\documentclass{article}

\usepackage{ijcai13}
\usepackage{times}
\usepackage{amsmath}
\usepackage{amsthm}
\usepackage{amsfonts}
\usepackage{graphicx}
\usepackage{algorithm}
\usepackage{algorithmic}

\newtheorem{proposition}{Proposition}
\newtheorem{theorem}{Theorem}
\newtheorem{corollary}{Corollary}

\title{A Consensual Linear Opinion Pool}

\author{
        Arthur Carvalho \\ 
        Cheriton School of Computer Science\\
        University of Waterloo\\
        a3carval@uwaterloo.ca 
        \And
        Kate Larson \\
        Cheriton School of Computer Science\\
        University of Waterloo\\ 
        klarson@uwaterloo.ca
        }

\begin{document}

\maketitle

\begin{abstract}

An important question when eliciting opinions from experts is how to aggregate the reported opinions. In this paper, we propose a pooling method to aggregate expert opinions. Intuitively, it works as if the experts were continuously updating their opinions in order to accommodate the expertise of others. Each updated opinion takes the form of a linear opinion pool, where the weight that an expert assigns to a peer's opinion is inversely related to the distance between their opinions. In other words, experts are assumed to prefer opinions that are close to their own opinions. We prove that such an updating process leads to consensus, \textit{i.e.}, the experts all converge towards the same opinion. Further, we show that if rational experts are rewarded using the quadratic scoring rule, then the assumption that they prefer opinions that are close to their own opinions follows naturally. We empirically demonstrate the efficacy of the proposed method using real-world data.

\end{abstract}

\section{Introduction}

Predicting outcomes of relevant uncertain events plays an essential role in decision-making processes. For example, companies rely on predictions about consumer demand and material supply to make their production plans, while weather forecasts provide guidelines for long range or seasonal agricultural planning, \textit{e.g.}, farmers can select crops that are best suited to the anticipated climatic conditions.

Forecasting techniques can be roughly divided into statistical and non-statistical methods. Statistical methods require historical data that contain valuable information about the future event. When such data are not available, a widely used non-statistical method is to request \emph{opinions} from experts regarding the future event~\cite{Cook:1991}. Opinions usually take the form of either numerical point estimates or probability distributions over plausible outcomes. We focus on opinions as probability mass functions.

The literature related to expert opinions is typically concerned about how expert opinions are used \cite{Mosleh:1988}, how uncertainty is or should be represented \cite{Ng:1990}, how experts do or should reason with uncertainty \cite{Cook:1991}, how to score the quality and usefulness of expert opinions \cite{Savage:1971,Boutilier:2012}, and how to produce a single consensual opinion when different experts report differing opinions \cite{DeGroot:1974}. It is this last question that we address in this paper.

We propose a pooling method to aggregate expert opinions that works as if the experts were continuously updating their opinions in order to accommodate the expertise and knowledge of others. Each updated opinion takes the form of a linear opinion pool, or a convex combination of opinions, where the weight that an expert assigns to a peer's opinion is inversely related to the distance between their opinions. In other words, experts are assumed to prefer opinions that are close to their own opinions. We prove that such an updating process leads to consensus, \textit{i.e.}, the experts all converge towards the same opinion. We also show that if the opinions of rational experts are scored using the quadratic scoring rule, then the assumption that experts prefer opinions that are close to their own follows naturally.

\section{Related Work}

The aggregation of expert opinions have been extensively studied in computer science and, in particular, artificial intelligence, \textit{e.g.}, the aggregation of opinions represented as preferences over a set of alternatives as in social choice theory \cite{Chevaleyre:2007}, the aggregation of point estimates using non-standard opinion pools \cite{Jurca:2008}, and the aggregation of probabilistic opinions using prediction markets \cite{Chen:2010}.

A traditional way of aggregating probabilistic opinions is through \emph{opinion pooling methods}. These methods are often divided into behavioral and mathematical methods \cite{Clemen:1999}. Behavioral aggregation methods attempt to generate agreement among the experts through interactions in order for them to share and exchange knowledge. Ideally, such sharing of information leads to a consensus. However, these methods typically provide no conditions under which the experts can be expected to reach agreement or even for terminating the iterative process.

On the other hand, mathematical aggregation methods consist of processes or analytical
models that operate on the individual probability distributions in order to produce a single, aggregate probability distribution. An important mathematical method is the \emph{linear opinion pool}, which involves taking a weighted linear average of the opinions \cite{Cook:1991}.

Several interpretations have been offered for the weights in the linear opinion pool. The performance-based approach recommends setting the weights based on previous performance of the experts \cite{Genest:1990}. A caveat with this approach is that performance measurements typically depend on the true outcome of the underlying event, which might not be available at the time when the opinions have to be aggregated. Also, previous successful (respective-ly, unsuccessful) predictions are not necessarily good indicators of future successful (respectively, unsuccessful) ones.

More closely related to this work is the interpretation of weights as a measure of distance. For example, Barlow \textit{et al.} \shortcite{Barlow:1986} proposed that the weight assigned to each expert's opinion should be inversely proportional to its distance to the most distant opinion, where distance is measured according to the Kullback-Leibler divergence. A clear drawback with this approach is that it only considers the distance to the most distant opinion when assigning a weight to an expert's opinion. Thus, even if the majority of experts have similar and accurate opinions, the weights of these experts' opinions in the aggregate prediction can be greatly reduced due to a single distant opinion.

For a comprehensive review of different perspectives on the weights in the linear opinion pool, we refer the interested reader to the work by Genest and McConway \shortcite{Genest:1990}.

\section{Model}

We consider the forecasting setting where a decision maker is interested in a probability vector over a set of mutually exclusive outcomes $\theta_1, \dots, \theta_z$, for $z \geq 2$. The decision maker deems it inappropriate to interject his own judgment about these outcomes. Hence, he elicits probabilistic opinions from $n$ experts. Experts' opinions are represented by $z$-dimensional probability vectors $\mathbf{f}_1, \dots, \mathbf{f}_n$. The probability vector $\mathbf{f}_i = (f_{i,1}, \dots, f_{i,z})$ represents expert $i$'s opinion, where  $f_{i, k}$ is his subjective probability regarding the occurrence of outcome $\theta_k$.

Since experts are not always in agreement, belief aggregation methods are used to combine their opinions into a single probability vector. Formally, $\mathbf{f} = T(\mathbf{f}_1, \dots, \mathbf{f}_n)$, where $\mathbf{f}$ is called an \emph{opinion pool}, and the function $T$ is the pooling operator. The \emph{linear opinion pool}  is a standard approach that involves taking a weighted linear average of the opinions:

\begin{equation}
\label{eq:linear_opinion_pool}
T(\mathbf{f}_1, \dots, \mathbf{f}_n) = \sum_{i=1}^n w_i \mathbf{f}_i
\end{equation}

\noindent where $w_i$ denotes the weight associated with expert $i$'s opinion. We make the standard assumption that $0 \leq w_i \leq 1$, for every $i \in \{1, \dots, n\}$, and $\sum_{i=1}^n w_i = 1$. 


\subsection{Consensus and Weights}

DeGroot \shortcite{DeGroot:1974} proposed a model which describes how a group can reach agreement on a common probability distribution by pooling their individual opinions. Initially, each expert $i$ is informed of the opinion of every other expert. In order to accommodate the information and expertise of the rest of the group, expert $i$ updates his own opinion as follows:

\begin{equation*}
\mathbf{f}_i^{(1)} = \sum_{j=1}^n p_{i,j}\mathbf{f}_j
\end{equation*}

\noindent where $p_{i,j}$ is the weight that expert $i$ assigns to the opinion of expert $j$ when he carries out this update. Weights must be chosen on the basis of the relative importance that experts assign to their peers' opinions.  It is assumed that $p_{i,j} > 0$, for every expert $i$ and $j$, and $\sum_{j=1}^n p_{i,j} = 1$. In this way, each updated opinion takes the form of a linear opinion pool. The whole updating process can be written in a slightly more general form using matrix notation, \textit{i.e.}, $\mathbf{F}^{(1)} = \mathbf{P}\mathbf{F}^{(0)}$, where:

\begin{align*}
\mathbf{P} &= \left[\begin{array}{cccc}
p_{1,1} & p_{1, 2} & \cdots & p_{1, n} \\
p_{2,1} & p_{2, 2} & \cdots & p_{2, n} \\
\vdots  & \vdots  & \ddots & \vdots  \\
p_{n,1} & p_{n, 2} & \cdots & p_{n, n} \\
\end{array}
\right], \qquad \mbox{and} \\
\mathbf{F}^{(0)} 
&=
\left[\begin{array}{c}
\mathbf{f}_{1} \\
\mathbf{f}_{2} \\
\vdots\\
\mathbf{f}_{n} \\
\end{array} 
\right]
=
\left[\begin{array}{cccc}
f_{1,1}  & f_{1, 2} & \cdots & f_{1, z} \\
f_{2,1}  & f_{2, 2} & \cdots & f_{2, z} \\
\vdots  & \vdots  & \ddots & \vdots  \\
f_{n,1}  & f_{n, 2} & \cdots & f_{n, z} \\
\end{array}
\right]
\end{align*}

Since all the opinions have changed, the experts might wish to revise their new opinions in the same way as they did before. If there is no basis for the experts to change their weights, we can then represent the whole updating process after $t$ revisions, for $t \geq 1$, as follows:

\begin{equation}
\label{eq:updating_opinions}
\mathbf{F}^{(t)} = \mathbf{P}\mathbf{F}^{(t-1)} = \mathbf{P}^{t}\mathbf{F}^{(0)}
\end{equation}

Let $\mathbf{f}^{(t)}_i = \left(f^{(t)}_{i,1}, \dots, f^{(t)}_{i,z}\right)$ be expert $i$'s opinion after $t$ updates, \textit{i.e.}, it denotes the $i$th row of the matrix $\mathbf{F}^{(t)}$. We say that a \emph{consensus} is reached if $\mathbf{f}^{(t)}_i = \mathbf{f}^{(t)}_j$, for every expert $i$ and $j$, as $t\rightarrow\infty$. Since $\mathbf{P}$, the matrix with weights, is a $n\times n$ stochastic matrix, it can then be regarded as the one-step transition probability matrix of a Markov chain with $n$ states and stationary probabilities. Consequently, one can apply a limit theorem that says that a consensus is reached when there exists a positive integer $t$ such that every element in at least one column of the matrix $\mathbf{P}^{t}$ is positive \cite{DeGroot:1974}.


\subsection{Weights as a Measure of Distance}

The original method proposed by DeGroot \shortcite{DeGroot:1974} has some drawbacks. First, the experts might want to change the weights that they assign to their peers' opinions after learning their initial opinions or after observing how much the opinions have changed from stage to stage. Further, opinions and/or identities have to be disclosed to the whole group when the experts are assigning the weights. Hence, privacy is not preserved, a fact which might be troublesome when the underlying event is of a sensitive nature.

In order to tackle these problems, we derive the weights that experts assign to the reported opinions by interpreting each weight as a measure of distance between two opinions.  We start by making the assumption that experts prefer opinions that are close to their own opinions, where closeness is measured by the following distance function:
\begin{equation}
\label{eq:root_mean_squared_deviation}
D(\mathbf{f}_i, \mathbf{f}_j) = \sqrt{\frac{\sum_{k=1}^{z} (f_{i,k} - f_{j,k})^2 }{z}}
\end{equation}
\noindent \textit{i.e.}, it is the root-mean-square deviation between two opinions $\mathbf{f}_i$ and $\mathbf{f}_j$. Given the above assumption, one can estimate the weight that expert $i$ assigns to expert $j$'s opinion at a given time $t$, for $t \geq 1$, as follows:
\begin{equation}
\label{eq:weight}
p_{i,j}^{(t)} = \frac{\alpha_i^{(t)}}{\epsilon + D\left(\mathbf{f}_i^{(t-1)}, \mathbf{f}_j^{(t-1)}\right)}
\end{equation}
\noindent where $\alpha_i^{(t)}$ normalizes the weights so that they sum to one, and $\epsilon$ is a small, positive constant used to avoid division by zero. We set $\mathbf{f}_i^{(0)} = \mathbf{f}_i$, \textit{i.e.}, it is the original opinion reported by expert $i$. There are some important points regarding equation (\ref{eq:weight}). First, the distance between two opinions is always non-negative. Hence, the constant $\epsilon$ ensures that every single weight is strictly greater than $0$ and strictly less than $1$. Further, the closer the opinions $\mathbf{f}_i^{(t-1)}$ and $\mathbf{f}_j^{(t-1)}$ are, the higher the resulting weight $p_{i,j}^{(t)}$ will be. Since $D\left(\mathbf{f}_i^{(t-1)}, \mathbf{f}_i^{(t-1)}\right) = 0$, the weight that each expert assigns to his own opinion is always greater than or equal to the weights that he assigns to his peers' opinions.

Now, we can redefine equation (\ref{eq:updating_opinions}) so as to allow the experts to update their weights based on the most recent opinions. After $t$ revisions, for $t \geq 1$, we have that $\mathbf{F}^{(t)} = \mathbf{P}^{(t)}\mathbf{F}^{(t-1)} = \mathbf{P}^{(t)}\mathbf{P}^{(t-1)} \dots \mathbf{P}^{(1)}\mathbf{F}^{(0)}$, where each element of each matrix  $\mathbf{P}^{(k)}$ is computed according to equation (\ref{eq:weight}):
\begin{align*}
\mathbf{P}^{(k)} &= \left[\begin{array}{cccc}
p_{1,1}^{(k)} & p_{1, 2}^{(k)} & \cdots & p_{1, n}^{(k)} \\
p_{2,1}^{(k)} & p_{2, 2}^{(k)} & \cdots & p_{2, n}^{(k)} \\
\vdots  & \vdots  & \ddots & \vdots  \\
p_{n,1}^{(k)} & p_{n, 2}^{(k)} & \cdots & p_{n, n}^{(k)} \\
\end{array}
\right]
\end{align*}
The opinion of each expert $i$ at time $t$ then becomes $\mathbf{f}_i^{(t)} = \sum_{j=1}^n p_{i,j}^{(t)}\mathbf{f}_j^{(t-1)}$. Algorithm 1 provides an algorithmic description of the proposed method.

In order to prove that all opinions converge towards a consensual opinion when using the proposed method, consider the following functions:
\begin{align*}
\delta\left(\mathbf{U}\right) &= \frac{1}{2}\max_{i,j}\sum_{k=1}^z \left| u_{i,k} - u_{j,k}\right| \\
\gamma(\mathbf{U}) &= \min_{i, j}\sum_{k=1}^z \min(u_{i, k}, u_{j, k})
\end{align*}

\noindent where $0 \leq \delta\left(\mathbf{U}\right), \gamma(\mathbf{U}) \leq 1$, and $\mathbf{U}$ is a stochastic matrix. $\delta\left(\mathbf{U}\right)$ computes the maximum absolute difference between two rows of a stochastic matrix $\mathbf{U}$. Thus, when $\delta\left(\mathbf{F}^{(t)}\right) = 0$, all rows of $\mathbf{F}^{(t)}$ are the same, \textit{i.e.}, a consensus is reached. We use the following results in our proof \cite{Paz:1971}:

\begin{algorithm}[H]
\caption{Algorithmic description of the proposed method to find a consensual opinion.}
\begin{algorithmic}[1] 
\REQUIRE $n$ probability vectors $\mathbf{f}_1^{(0)}, \dots, \mathbf{f}_n^{(0)}$.
\REQUIRE recalibration factor $\epsilon$.
\FOR{$ t = 1 \textbf{\textrm{ to }} \infty$}  
\FOR{$ i = 1 \textbf{\textrm{ to }} n$} 
\FOR{$ j = 1 \textbf{\textrm{ to }} n$} 
\STATE $p_{i,j}^{(t)} = \frac{\alpha_i^{(t)}}{\epsilon + D\left(\mathbf{f}_i^{(t-1)}, \mathbf{f}_j^{(t-1)}\right)}$  
\ENDFOR
\STATE $\mathbf{f}_i^{(t)} = \sum_{j=1}^n p_{i,j}^{(t)}\mathbf{f}_j^{(t-1)}$
\ENDFOR
\ENDFOR
\end{algorithmic}
\end{algorithm}

\begin{proposition}
Given two stochastic matrices $\mathbf{U}$ and $\mathbf{V}$, $\delta(\mathbf{UV}) \leq \delta(\mathbf{U})\delta(\mathbf{V})$.
\end{proposition}

\begin{proposition}
Given a stochastic matrix $\mathbf{U}$, then $\delta(\mathbf{U}) = 1 - \gamma(\mathbf{U})$.
\end{proposition}

Our main result is stated below.

\begin{theorem}
When $t \rightarrow \infty$, $\mathbf{f}^{(t)}_i = \mathbf{f}^{(t)}_j$, for every expert $i$ and $j$.
\end{theorem}

\begin{proof}
Recall that $\mathbf{F}^{(t)}$ is the stochastic matrix representing the experts' opinions after $t$ revisions, and that $\mathbf{F}^{(t)} = \mathbf{P}^{(t)}\mathbf{F}^{(t-1)}$. Now, consider the following sequence: $\left(\delta\left(\mathbf{F}^{(0)}\right), \delta\left(\mathbf{F}^{(1)}\right), ..., \delta\left(\mathbf{F}^{(t)}\right)\right)$. We are interested in the behavior of this sequence when $t \rightarrow \infty$. First, we show that such a sequence is monotonically decreasing:

\begin{align*}
\delta\left(\mathbf{F}^{(t)}\right) &= \delta\left(\mathbf{P}^{(t)}\mathbf{F}^{(t-1)}\right)\\ 
& \leq \delta\left(\mathbf{P}^{(t)}\right)\delta\left(\mathbf{F}^{(t-1)}\right)\\
&= \left(1 - \gamma\left(\mathbf{P}^{(t)}\right)\right) \delta\left(\mathbf{F}^{(t-1)}\right)\\
&\leq \delta\left(\mathbf{F}^{(t-1)}\right)
\end{align*}

The second and third lines follow, respectively, from Propositions 1 and 2. Since $\delta\left(\mathbf{U}\right) \geq 0$ for every stochastic matrix $\mathbf{U}$, then the above mentioned sequence is a bounded decreasing sequence. Hence, we can apply the standard monotone convergence theorem \cite{Bartle:2000} and $\delta\left(\mathbf{F}^{(\infty)}\right) = 0$. Consequently, all rows of the stochastic matrix $\mathbf{F}^{(\infty)}$ are the same.
\end{proof}

In other words, a consensus is always reached under the proposed method, and this does not depend on the initial reported opinions. A straightforward corollary of Theorem 1 is that all revised weights converge to the same value.

\begin{corollary}
When $t \rightarrow \infty,  p_{i,j}^{(t)} = \frac{1}{n}, $ for every expert $i$ and $j$.
\end{corollary}

Hence, the proposed method works as if experts were continuously exchanging information so that their individual knowledge becomes group knowledge and all opinions are equally weighted. Since we derive weights from the reported opinions, we are then able to avoid some problems that might arise when eliciting these weights directly, \textit{e.g.}, opinions do not need to be disclosed to others in order for them to assign weights, thus preserving privacy.

The resulting consensual opinion can be represented as an instance of the linear opinion pool. Recall that $\mathbf{f}_i^{(t)} = \sum_{j=1}^n p_{i,j}^{(t)}\mathbf{f}_j^{(t-1)} = \sum_{j=1}^n p_{i,j}^{(t)}\sum_{k=1}^n p_{j,k}^{(t-1)} \mathbf{f}_k^{(t-2)} = \dots =  \sum_{j=1}^n \beta_j\mathbf{f}_j^{(0)}$, where $\mathbf{\beta} = (\beta_1, \beta_2, \dots, \beta_n)$ is a probability vector that incorporates all the previous weights. Hence, another interpretation of the proposed method is that experts reach a consensus regarding the weights in equation (\ref{eq:linear_opinion_pool}).

\subsection{Numerical Example}

A numerical example may clarify the mechanics of the proposed method. Consider three experts ($n=3$) with the following opinions: $\mathbf{f}_1 = (0.9, 0.1)$, $\mathbf{f}_2 = (0.05, 0.95)$, and $\mathbf{f}_3 = (0.2, 0.8)$. According to (\ref{eq:root_mean_squared_deviation}), the initial distance between, say, $\mathbf{f}_1$ and $\mathbf{f}_2$ is:

\begin{displaymath}
D(\mathbf{f}_1, \mathbf{f}_2) = \sqrt{\frac{(0.9 - 0.05)^2 + (0.1-0.95)^2}{2}} = 0.85
\end{displaymath}

Similarly, we have that $D(\mathbf{f}_1, \mathbf{f}_1) = 0$ and $D(\mathbf{f}_1, \mathbf{f}_3)  = 0.7$. Using equation (\ref{eq:weight}), we can then derive the weights that each expert assigns to the reported opinions. Focusing on expert $1$ at time $t = 1$ and setting $\epsilon = 0.01$, we obtain $p_{1,1}^{(1)} = \alpha_1^{(1)}/0.01$, $p_{1,2}^{(1)} = \alpha_1^{(1)}/0.86$, and $p_{1,3}^{(1)} = \alpha_1^{(1)}/0.71$. Since these weights must sum to one, we have $\alpha_1^{(1)} \approx 0.00975$ and, consequently, $p_{1,1}^{(1)} \approx 0.975$, $p_{1,2}^{(1)} \approx 0.011$, and $p_{1,3}^{(1)} \approx 0.014$. Repeating the same procedure for all experts, we obtain the matrix:

\begin{align*}
\mathbf{P}^{(1)} &= \left[\begin{array}{cccc}
0.975 & 0.011 & 0.014\\
0.011 & 0.931 & 0.058\\
0.013 & 0.058 & 0.929
\end{array}
\right]
\end{align*}

The updated opinion of expert $1$  is then $\mathbf{f}_1^{(1)} = \sum_{j=1}^3 p_{1,j}^{(1)}\mathbf{f}_j \approx (0.8809, 0.1191)$. By repeating the above procedure, when $t \rightarrow \infty$, $\mathbf{P}^{(t)}$ converges to a matrix where all the elements are equal to $1/3$. Moreover, all experts' opinions converge to the prediction $(0.3175, 0.6825)$. An interesting point to note is that the resulting prediction would be $(0.3833, 0.6167)$ if we had taken the average of the reported opinions, \textit{i.e.}, expert 1, who has a very different opinion, would have more influence on the aggregate prediction.

\section{Consensus and Proper Scoring Rules}

The major assumption of the proposed method is that experts prefer opinions that are close to their own opinions. In this section, we formally investigate the validity of this assumption. We start by noting that in the absence of a well-chosen incentive structure, the experts might indulge in game playing which distorts their reported opinions. For example, experts who have a reputation to protect might tend to produce forecasts near the most likely group consensus, whereas experts who have a reputation to build might tend to overstate the probabilities of outcomes they feel will be understated in a possible consensus~\cite{Friedman:1983}.

\emph{Scoring rules} are traditional devices used to promote honesty in forecasting settings \cite{Savage:1971}. Formally, a scoring rule is a real-valued function, $R(\mathbf{f}_i, e)$, that provides a score for the opinion $\mathbf{f}_i$ upon observing the outcome $\theta_e$. Given that experts' utility functions are linear with respect to the range of the score used in conjunction with the scoring rule, the condition that $R$ is \emph{strictly proper} implies that the opinion reported by each expert strictly maximizes his expected utility if and only if he is honest. Formally, $\mbox{argmax}_{\mathbf{f}_i^\prime} \mathbb{E}_{\mathbf{f}_i} \left[R(\mathbf{f}_i^\prime)\right] = \mathbf{f}_i$, where $\mathbb{E}_{\mathbf{f}_i} \left[R(\cdot)\right]$ is the $\mathbf{f}_i$-expected value of $R$, $\mathbb{E}_{\mathbf{f}_i} \left[R(\mathbf{f}_i^\prime)\right] = \sum_{e = 1}^z {f}_{i, e}\, R(\mathbf{f}_i^\prime, e)$. A well-known strictly proper scoring rule is the \emph{quadratic scoring rule}:
\begin{equation}
\label{eq:quadratic_scoring_rule}
R(\mathbf{f}_i, e) = 2f_{i, e} - \sum_{k=1}^z f_{i, k}^2
\end{equation}
The proof that the quadratic scoring rule is indeed strictly proper as well as some of its interesting properties can be seen in the work by Selten \shortcite{Selten:1998}.

Proper scoring rules have been used as a tool to promote truthfulness in a variety of domains, \textit{e.g.}, when sharing rewards among a set of agents based on peer evaluations \cite{Carvalho:2010,Carvalho:2011,Carvalho:2012}, to incentivize agents to accurately estimate their own efforts to accomplish a task \cite{Bacon:2012}, in prediction markets \cite{Hanson:2003}, in weather forecasting \cite{Gneiting:2007}, \textit{etc}.

\subsection{Effective Scoring Rules}

Scoring rules can also be classified based on monotonicity properties. Consider a metric $G$ that assigns to any pair of opinions $\mathbf{f}_i$ and $\mathbf{f}_j$ a real number, which in turn can be seen as the shortest distance between $\mathbf{f}_i$ and $\mathbf{f}_j$. We say that a scoring rule $R$ is \emph{effective} with respect to $G$ if the following relation holds for any opinions $\mathbf{f}_i, \mathbf{f}_j$, and $\mathbf{f}_k$~\cite{Friedman:1983}:
\begin{equation*}
G(\mathbf{f}_i, \mathbf{f}_j) < G(\mathbf{f}_i, \mathbf{f}_k)
\iff
\mathbb{E}_{\mathbf{f}_i} \left[R(\mathbf{f}_j)\right] > \mathbb{E}_{\mathbf{f}_i} \left[R(\mathbf{f}_k)\right] 
\end{equation*}
In words, each expert's expected score can be seen as a monotone decreasing function of the distance between his true opinion and the reported one, \textit{i.e.}, experts still strictly maximize their expected scores by telling the truth, and the closer a reported opinion is to the true opinion, the higher the expected score will be. The property of effectiveness is stronger than strict properness, and it has been proposed as a desideratum for scoring rules for reasons of monotonicity in keeping an expert close to his true opinion \cite{Friedman:1983}.

By definition, a metric $G$ must satisfy the following conditions for any opinions $\mathbf{f}_i, \mathbf{f}_j$, and $\mathbf{f}_k$:

\begin{enumerate}
\item Positivity: $G(\mathbf{f}_i, \mathbf{f}_j) \geq 0$, for all experts $i, j$, and $G(\mathbf{f}_i, \mathbf{f}_j) = 0$ if and only if $\mathbf{f}_i = \mathbf{f}_j$;
\item Symmetry: $G(\mathbf{f}_i, \mathbf{f}_j) = G(\mathbf{f}_j, \mathbf{f}_i)$;
\item Triangle Inequality: $G(\mathbf{f}_i, \mathbf{f}_k) \leq G(\mathbf{f}_i, \mathbf{f}_j) + G(\mathbf{f}_j, \mathbf{f}_k)$.
\end{enumerate}

The root-mean-square deviation shown in (\ref{eq:root_mean_squared_deviation}) satisfies the above conditions. However, equation (\ref{eq:weight}), taken as a function of opinions, is not a true metric, \textit{e.g.}, symmetry does not always hold. We adjust the original definition of effective scoring rules so as to consider weights instead of metrics. We say that a scoring rule $R$ is effective with respect to a set of weights $W = \{p_{1,1}^{(t)}, \dots, p_{1,n}^{(t)}, p_{2,1}^{(t)}, \dots, p_{n,n}^{(t)}\}$ assigned at any time $t \geq 1$ if the following relation holds for any opinions $\mathbf{f}_i^{(t-1)}, \mathbf{f}_j^{(t-1)}$, and $\mathbf{f}_k^{(t-1)}$:

\begin{equation*}
p_{i,j}^{(t)} < p_{i,k}^{(t)} \iff 
\mathbb{E}_{\mathbf{f}_i^{(t-1)}} [R (\mathbf{f}_k^{(t-1)})] > \mathbb{E}_{\mathbf{f}_i^{(t-1)}} [R (\mathbf{f}_j^{(t-1)})]
\end{equation*}

In words, each expert's expected score can be seen as a monotone increasing function of his assigned weights, \textit{i.e.}, the higher the weight one expert assigns to a peer's opinion, the greater the expected score of that expert would be if he reported his peer's opinion, and vice versa. We prove below that the quadratic scoring rule shown in (\ref{eq:quadratic_scoring_rule}) is effective with respect to a set of weights assigned according to (\ref{eq:weight}).

\begin{proposition}
The quadratic scoring rule shown in (\ref{eq:quadratic_scoring_rule}) is effective with respect to a set of weights $W = \{p_{1,1}^{(t)}, \dots, p_{1,n}^{(t)}, p_{2,1}^{(t)}, \dots, p_{n,n}^{(t)}\}$ assigned at any time $t \geq 1$ according to equation (\ref{eq:weight}). 
\end{proposition}

\begin{proof}

Given an opinion $\mathbf{f}_j$, we note that the $\mathbf{f}_i$-expected value of the quadratic scoring rule in (\ref{eq:quadratic_scoring_rule}) can be written as:
\begin{align*}
\mathbb{E}_{\mathbf{f}_i} \left[R(\mathbf{f}_j)\right] 
&=
\sum_{e = 1}^z {f}_{i, e}\, R(\mathbf{f}_j, e) \\
&=
\sum_{e = 1}^z \left(2{f}_{j, e}{f}_{i, e} - {f}_{i, e}\sum_{x=1}^z f_{j, x}^2\right)\\
&=
\sum_{e = 1}^z 2{f}_{j, e}{f}_{i, e} - \sum_{e = 1}^z{f}_{i, e}\sum_{x=1}^z f_{j, x}^2\\
&=
\sum_{e = 1}^z 2{f}_{j, e}{f}_{i, e} - \sum_{x=1}^z f_{j, x}^2\label{EQUATION:proposition}\\
\end{align*}

Now, consider the weights assigned by expert $i$ to the opinions of experts $j$ and $k$ at time $t \geq 1$ according to equation (\ref{eq:weight}). We have that $p_{i,j}^{(t)} < p_{i,k}^{(t)}$ if and only if:

\begin{align*}
&\frac{\alpha_i^{(t)}}{\epsilon+D\left(\mathbf{f}_i^{(t-1)}, \mathbf{f}_j^{(t-1)}\right)} <  \frac{\alpha_i^{(t)}}{\epsilon+D\left(\mathbf{f}_i^{(t-1)}, \mathbf{f}_k^{(t-1)}\right)} &\equiv \\
&D\left(\mathbf{f}_i^{(t-1)}, \mathbf{f}_k^{(t-1)}\right) < D\left(\mathbf{f}_i^{(t-1)}, \mathbf{f}_j^{(t-1)}\right) 
&\equiv\\
&\sum_{x=1}^{z} \left(f_{i,x}^{(t-1)} - f_{k,x}^{(t-1)}\right)^2 < \sum_{x=1}^{z} \left(f_{i,x}^{(t-1)} - f_{j,x}^{(t-1)}\right)^2
&\equiv\\
&\sum_{x=1}^{z} 2f_{i,x}^{(t-1)}f_{k,x}^{(t-1)} -\sum_{y=1}^{z} \left(f_{k,y}^{(t-1)}\right)^2 > 
& \\ 
&\sum_{x=1}^{z} 2f_{i,x}^{(t-1)}f_{j,x}^{(t-1)} -\sum_{y=1}^{z} \left(f_{j,y}^{(t-1)}\right)^2 
&\equiv\\
&\mathbb{E}_{\mathbf{f}_i^{(t-1)}} \left[R\left(\mathbf{f}_k^{(t-1)}\right)\right] > \mathbb{E}_{\mathbf{f}_i^{(t-1)}} \left[R\left(\mathbf{f}_j^{(t-1)}\right)\right]
\end{align*}\end{proof}

Proposition 3 implies that there is a correspondence between weights, assigned according to (\ref{eq:weight}), and expected scores from the quadratic scoring rule: the higher the weight one expert assigns to a peer's opinion, the greater that expert's expected score would be if he reported his peer's opinion, and vice versa. Hence, whenever experts are rational, \textit{i.e.}, when they behave so as to maximize their expected scores, and their opinions are rewarded using the quadratic scoring rule, then the major assumption of the proposed method for finding a consensual opinion, namely that experts prefer opinions that are close to their own opinions, is formally valid. A straightforward corollary of Proposition 3 is that a positive affine transformation of the quadratic scoring rule is still effective with respect to a set of weights assigned  according to (\ref{eq:weight}).

\begin{corollary}
A positive affine transformation of the quadratic scoring rule $R$ in (\ref{eq:quadratic_scoring_rule}), \textit{i.e.}, $x R\left(\mathbf{f}_i, e\right) + y$, for $x > 0$ and $y \in \Re$,  is effective with respect to a set of weights $W = \{p_{1,1}^{(t)}, \dots, p_{1,n}^{(t)}, p_{2,1}^{(t)}, \dots, p_{n,n}^{(t)}\}$ assigned at any time $t \geq 1$ according to equation (\ref{eq:weight}). 
\end{corollary}

\section{Empirical Evaluation}

In this section, we describe an experiment designed to test the efficacy of the proposed method for finding a consensual opinion. In the following subsections, we describe the dataset used in our experiments, the metrics used to compare different methods to aggregate opinions, and the obtained results.

\subsection{Dataset}

Our dataset was composed by 267 games (256 regular-season games and 11 playoff games) from the National Football League (NFL) held between September 8th, 2005 and February 5th, 2006. We obtained the opinions of 519 experts for the NFL games from the ProbabilityFootball\footnote{Available at http://probabilityfootball.com/2005/} contest.  The contest was free to enter. Each expert was asked to report his subjective probability that a team would win a game. Predictions had to be reported by noon on the day of the game. Since the probability of a tie in NFL games is very low (less than $1\%$), experts did not report the probability of such an outcome. In particular, no ties occurred in our dataset.

Not all 519 registered experts reported their predictions for every game. An expert who did not enter a prediction for a game was removed from the opinion pool for that game. On average, each game attracted approximately 432 experts, the standard deviation being equal to 26.37. The minimum and maximum number of experts were, respectively, 243 and 462. Importantly, the contest rewarded the performance of experts via a positive affine transformation of the quadratic scoring rule, \textit{i.e.}, $100 - 400 \times p_l^2$, where $p_l$ was the probability that an expert assigned to the eventual losing team.

A positive affine transformation of a strictly proper scoring rule is still strictly proper \cite{Gneiting:2007}. The above scoring rule can be obtained by multiplying (\ref{eq:quadratic_scoring_rule}) by $200$ and subtracting the result by $100$. The resulting proper scoring rule rewards bold predictions more when they are right. Likewise, it  penalizes bold predictions more when they are wrong. For example, a prediction of $99\%$ earns $99.96$ points if the chosen team wins, and it loses $292.04$ points if the chosen team loses. On the other hand, a prediction of $51\%$ earns $3.96$ points if it is correct, and it loses $4.04$ points if it is wrong. A prediction of $50\%$ neither gains nor loses any points. The experts with highest accumulated scores won prizes in the contest. The suggested strategy at the contest website was ``\textit{to make picks for each game that match, as closely as possible, the probabilities that each team will win}".

We argue that this dataset is very suitable for our purposes due to many reasons. First, the popularity of NFL games provides natural incentives for people to participate in the ProbabilityFootball contest. Furthermore, the intense media coverage and scrutiny of the strengths and weaknesses of the teams and individual players provide useful information for the general public. Hence, participants of the contest can be viewed as knowledgeable regarding to the forecasting goal. Finally, the fact that experts were rewarded via a positive affine transformation of the quadratic scoring rule fits perfectly into the theory developed in this work (see Corollary 2).

\subsection{Metrics}

We used two different metrics to assess the prediction power of different aggregation methods.

\subsubsection*{Overall Accuracy}

We say that a team is the predicted favorite for winning a game when an aggregate prediction that this team will win the game is greater than $0.5$. Overall accuracy is then the percentage of games that predicted favorites have indeed won. A polling method  with higher overall accuracy is more accurate.

\subsubsection*{Absolute Error}

Absolute error is the difference between a perfect prediction (1 for the winning team) and the actual prediction. Thus, it is just the probability assigned to the losing team ($p_l$). An aggregate prediction with lower absolute error is more accurate.

\subsection{Experimental Results}

For each game in our dataset, we aggregated the reported opinions using three different linear opinion pools: the method proposed in Section 3, henceforth referred to as the \emph{consensual} method, with $\epsilon = 10^{-4}$; the traditional \emph{average} approach, where all the weights in  (\ref{eq:linear_opinion_pool}) are equal to $1/n$; and the method proposed by Barlow \textit{et al.} \shortcite{Barlow:1986}, henceforth referred to as the \emph{BMS} method. These authors proposed that the weight assigned to expert $i$'s opinion should be $w_i = \frac{c}{I(\mathbf{f}_i, \mathbf{f}_{i^*})}$, where $c$ is a normalizing constant, $I(\mathbf{f}_i, \mathbf{f}_{i^*})$ is the Kullback-Leibler divergence, and $\mathbf{f}_{i^*}$ achieves $\max\{I(\mathbf{f}_i, \mathbf{f}_j): 1 \leq j \leq n \}$, \textit{i.e.}, $\mathbf{f}_{i^*}$ is the most distant opinion from expert $i$'s opinion. The BMS method produces indeterminate outputs whenever there are probability assessments equal to 0 or 1. Hence, we recalibrated the reported opinions when using the BMS method by replacing 0 and 1 by, respectively, 0.01 and 0.99.

Given the aggregated opinions, we calculated the performance of each method according to the accuracy metrics previously described. Regarding the overall accuracy of each method, the consensual method achieves the best performance in this experiment with an overall accuracy of $69.29\%$. The BMS and average methods achieve an overall accuracy of, respectively,  $68.54\%$ and $67.42\%$.

Table \ref{table:absolute_error} shows the average absolute error of each method over the 267 games. The consensual method achieves the best performance with an average absolute error of $0.4115$. 
We performed left-tailed  Wilcoxon signed-rank tests in order to investigate the statistical relevance of these results. The resulting p-values are all extremely small $\left(< 10^{-4}\right)$, showing that the results are indeed statistically significant.

\begin{table}[t]
\centering
\caption{\textit{Average absolute error of each method over the 267 games. Standard deviations are in parentheses.}}
\label{table:absolute_error}
\begin{tabular}{ccc} 
\hline\noalign{\smallskip}
Consensual & Average & BMS\\
\noalign{\smallskip}\hline\noalign{\smallskip}
0.4115  (0.1813)  & 0.4176 (0.1684)   &  0.4295 (0.1438)\\
\noalign{\smallskip}\hline
\end{tabular}
\end{table}

Despite displaying a decent overall accuracy, the BMS method has the worst performance according to the absolute error metric. A clear drawback with this method is that it only considers the distance to the most distant opinion when assigning a weight to an opinion. Since our experimental design involves hundreds of experts, it is reasonable to expect at least one of them to have a very different and wrong opinion.

The high number of experts should give an advantage to the average method since biases of individual judgment can offset with each other when opinions are diverse, thus making the aggregate prediction more accurate. However, the average method achieves the worst overall accuracy, and it performs statistically worse than the consensual method when measured under the absolute error metric. We believe this result happens because the average method ends up overweighting extreme opinions when equally weighting all opinions.

On the other hand, under the consensual method, experts put less weight on opinions far from their own opinions, which implies that this method is generally less influenced by extreme predictions as illustrated in Section 3.3.

\section{Conclusion}

We proposed a pooling method to aggregate expert opinions. Intuitively, the proposed method works as if the experts were continuously updating their opinions, where each updated opinion takes the form of a linear opinion pool, and the weight that each expert assigns to a peer's opinion is inversely related to the distance between their opinions. We proved that this updating process leads to a consensus.

A different interpretation of the proposed method is that experts reach a consensus regarding the weights of a linear opinion pool.  We showed that if rational experts are rewarded using the quadratic scoring rule, then our major assumption, namely that experts prefer opinions that are close to their own opinions, follows naturally.  To the best of our knowledge, this is the first work linking the theory of proper scoring rules to the seminal consensus theory proposed by DeGroot \shortcite{DeGroot:1974}.

Using real-world data, we compared the performance of the proposed method with two other methods: the traditional average approach and another distance-based aggregation method proposed by Barlow \textit{et al.} \shortcite{Barlow:1986}. The results of our experiment show that the proposed method outperforms all the other methods when measured in terms of both overall accuracy and absolute error.

\bibliographystyle{named}
\bibliography{bibliography}

\begin{thebibliography}{}

\bibitem[\protect\citeauthoryear{Bacon \bgroup \em et al.\egroup
  }{2012}]{Bacon:2012}
David~F. Bacon, Yiling Chen, Ian Kash, David~C. Parkes, Malvika Rao, and Manu
  Sridharan.
\newblock {Predicting Your Own Effort}.
\newblock In {\em Proceedings of the 11th International Conference on
  Autonomous Agents and Multiagent Systems}, pages 695--702, 2012.

\bibitem[\protect\citeauthoryear{Barlow \bgroup \em et al.\egroup
  }{1986}]{Barlow:1986}
R.~E. Barlow, R.~W. Mensing, and N.~G. Smiriga.
\newblock Combination of experts' opinions based on decision theory.
\newblock In A.~P. Basu, editor, {\em Reliability and quality control}, pages
  9--19. North-Holland, 1986.

\bibitem[\protect\citeauthoryear{Bartle and Sherbert}{2000}]{Bartle:2000}
R.~G. Bartle and D.~R. Sherbert.
\newblock {\em Introduction to Real Analysis}.
\newblock Wiley, 3rd edition, 2000.

\bibitem[\protect\citeauthoryear{Boutilier}{2012}]{Boutilier:2012}
C.~Boutilier.
\newblock Eliciting forecasts from self-interested experts: Scoring rules for
  decision makers.
\newblock In {\em Proceedings of the 11th International Conference on
  Autonomous Agents and Multiagent Systems}, pages 737--744, 2012.

\bibitem[\protect\citeauthoryear{Carvalho and Larson}{2010}]{Carvalho:2010}
Arthur Carvalho and Kate Larson.
\newblock {Sharing a reward based on peer evaluations}.
\newblock In {\em Proceedings of the 9th International Conference on Autonomous
  Agents and Multiagent Systems}, pages 1455--1456, 2010.

\bibitem[\protect\citeauthoryear{Carvalho and Larson}{2011}]{Carvalho:2011}
Arthur Carvalho and Kate Larson.
\newblock {A truth serum for sharing rewards}.
\newblock In {\em Proceedings of the 10th International Conference on
  Autonomous Agents and Multiagent Systems}, pages 635--642, 2011.

\bibitem[\protect\citeauthoryear{Carvalho and Larson}{2012}]{Carvalho:2012}
Arthur Carvalho and Kate Larson.
\newblock {Sharing rewards among strangers based on peer evaluations}.
\newblock {\em Decision Analysis}, 9(3):253--273, 2012.

\bibitem[\protect\citeauthoryear{Chen and Pennock}{2010}]{Chen:2010}
Y.~Chen and D.M. Pennock.
\newblock Designing markets for prediction.
\newblock {\em AI Magazine}, 31(4):42--52, 2010.

\bibitem[\protect\citeauthoryear{Chevaleyre \bgroup \em et al.\egroup
  }{2007}]{Chevaleyre:2007}
Y.~Chevaleyre, U.~Endriss, J.~Lang, and N.~Maudet.
\newblock A short introduction to computational social choice.
\newblock In {\em Proceedings of the 33rd conference on Current Trends in
  Theory and Practice of Computer Science}, pages 51--69, 2007.

\bibitem[\protect\citeauthoryear{Clemen and Winkler}{1999}]{Clemen:1999}
R.T. Clemen and R.L. Winkler.
\newblock Combining probability distributions from experts in risk analysis.
\newblock {\em Risk Analysis}, 19:187--203, 1999.

\bibitem[\protect\citeauthoryear{Cooke}{1991}]{Cook:1991}
R.M. Cooke.
\newblock {\em Experts in uncertainty: opinion and subjective probability in
  science}.
\newblock Oxford University Press, 1991.

\bibitem[\protect\citeauthoryear{DeGroot}{1974}]{DeGroot:1974}
M.H. DeGroot.
\newblock Reaching a consensus.
\newblock {\em Journal of the American Statistical Association},
  69(345):118--121, 1974.

\bibitem[\protect\citeauthoryear{Friedman}{1983}]{Friedman:1983}
D.~Friedman.
\newblock Effective scoring rules for probabilistic forecasts.
\newblock {\em Management Science}, 29(4):447--454, 1983.

\bibitem[\protect\citeauthoryear{Genest and McConway}{1990}]{Genest:1990}
C.~Genest and K.J. McConway.
\newblock Allocating the weights in the linear opinion pool.
\newblock {\em Journal of Forecasting}, 9(1):53--73, 1990.

\bibitem[\protect\citeauthoryear{Gneiting and Raftery}{2007}]{Gneiting:2007}
Tilmann Gneiting and Adrian~E. Raftery.
\newblock Strictly proper scoring rules, prediction, and estimation.
\newblock {\em Journal of the American Statistical Association},
  102(477):359--378, 2007.

\bibitem[\protect\citeauthoryear{Hanson}{2003}]{Hanson:2003}
R.~Hanson.
\newblock Combinatorial information market design.
\newblock {\em Information Systems Frontiers}, 5(1):107--119, 2003.

\bibitem[\protect\citeauthoryear{Jurca and Faltings}{2008}]{Jurca:2008}
R.~Jurca and B.~Faltings.
\newblock Incentives for expressing opinions in online polls.
\newblock In {\em Proceeddings of the 2008 ACM Conference on Electronic
  Commerce}, pages 119--128, 2008.

\bibitem[\protect\citeauthoryear{Mosleh \bgroup \em et al.\egroup
  }{1988}]{Mosleh:1988}
A.~Mosleh, V.M. Bier, and G.~Apostolakis.
\newblock A critique of current practice for the use of expert opinions in
  probabilistic risk assessment.
\newblock {\em Reliability Engineering \& System Safety}, 20(1):63 -- 85, 1988.

\bibitem[\protect\citeauthoryear{Ng and Abramson}{1990}]{Ng:1990}
K.C. Ng and B.~Abramson.
\newblock Uncertainty management in expert systems.
\newblock {\em IEEE Expert}, 5(2):29 -- 48, april 1990.

\bibitem[\protect\citeauthoryear{Paz}{1971}]{Paz:1971}
A.~Paz.
\newblock {\em Introduction to Probabilistic Automata}.
\newblock Academic Press, 1971.

\bibitem[\protect\citeauthoryear{Savage}{1971}]{Savage:1971}
L.J. Savage.
\newblock {Elicitation of Personal Probabilities and Expectations}.
\newblock {\em Journal of the American Statistical Association},
  66(336):783--801, 1971.

\bibitem[\protect\citeauthoryear{Selten}{1998}]{Selten:1998}
R.~Selten.
\newblock Axiomatic characterization of the quadratic scoring rule.
\newblock {\em Experimental Economics}, 1(1):43--62, 1998.

\end{thebibliography}

\end{document}